\documentclass[journal, onecolumn]{IEEEtran}
\IEEEoverridecommandlockouts
\usepackage{microtype}
\usepackage{graphicx}
\usepackage{booktabs} % for professional tables
\usepackage{environ}% 
\usepackage[CJKbookmarks=true,
            bookmarksnumbered=true,
            bookmarksopen=true,
            colorlinks=true,
            citecolor=blue,
            linkcolor=red,
            anchorcolor=red,
            urlcolor=blue
            ]{hyperref}
\usepackage{geometry}
\usepackage{amsmath}
\usepackage{amssymb}
\usepackage{mathtools}
\usepackage{amsthm}
\usepackage{multicol}
\usepackage{multirow}
\usepackage{dsfont}
\usepackage[capitalize,noabbrev]{cleveref}
\usepackage{babel}
\usepackage{verbatim}
\usepackage{mathtools}
\usepackage{bm}
\usepackage{natbib}
\usepackage{amsmath}
\usepackage{amssymb}
\usepackage{multicol}
\usepackage{algorithm}
\usepackage{algorithmicx}
\usepackage[noend]{algpseudocode}
\usepackage{enumitem}
\usepackage{caption}
\usepackage{subcaption}
\usepackage{bbm}
\usepackage{color}
\usepackage[normalem]{ulem}

\usepackage[textsize=tiny]{todonotes}

\theoremstyle{plain}
\newtheorem{theorem}{Theorem}[section]

\newtheorem{lemma}[theorem]{Lemma}

\theoremstyle{definition}

\newtheorem{remark}{Remark}[section]

\theoremstyle{remark}

\setcitestyle{authoryear,open={(},close={)}} 

\definecolor{cm}{RGB}{0,0,200}
\definecolor{purple}{RGB}{200,0,200}
\def\cs{{\mathcal S}}
\def\ct{{\mathcal T}}
\def\ce{{\mathcal E}}
\def\ca{{\mathcal A}}
\def\cb{{\mathcal B}}
\def\cx{{\mathcal X}}
\def\cy{{\mathcal Y}}
\def\co{{\mathcal O}}
\DeclareMathOperator\E{\mathbb{E}}
\let\P\relax
\DeclareMathOperator\P{\mathsf{P}}
\def\DKL{{D_\mathsf{KL}(p\|1-p)}}
\newcommand{\cond}{\mathchoice{\,\vert\,}{\mspace{2mu}\vert\mspace{2mu}}{\vert}{\vert}}
\newcommand{\snc}{\textsc{NoisyCompare}}

\newcommand{\sno}{\textsc{NoisyOR }}
\newcommand{\scb}{\textsc{CheckBit }}
\newcommand{\stourmax}{\textsc{Tournament\_Max }}
\newcommand{\stouror}{\textsc{Tournament\_OR }}

\makeatletter
\newcommand{\orn}{\ensuremath{\mathsf{OR}} }
\newcommand{\maxn}{\ensuremath{\mathsf{MAX}} }

\DeclareMathOperator*{\argmax}{arg\,max}

\title{Noisy Computing of the \orn \\ and \maxn Functions}

\author{Banghua Zhu$^*$,
\and Ziao Wang$^*$,
\and Nadim Ghaddar$^*$,
\and Jiantao Jiao,
\and Lele Wang
\thanks{$*$ Banghua Zhu, Ziao Wang and Nadim Ghaddar contributed equally to this work.}
\thanks{Banghua Zhu is with the Department of Electrical Engineering and Computer Sciences, University of California Berkeley, Berkeley, CA 94720, USA, (email: banghua@berkeley.edu).}
\thanks{Ziao Wang is with the Department of Electrical and Computer Engineering, University of British Columbia, Vancouver, BC V6T1Z4, Canada (email: ziaow@ece.ubc.ca).}
\thanks{Nadim Ghaddar is with the Department of Electrical and Computer Engineering, University of California San Diego, La Jolla, CA 92093, USA, (email: nghaddar@ucsd.edu).}
\thanks{Jiantao Jiao is with the Department of Electrical Engineering and Computer Sciences, University of California Berkeley, Berkeley, CA 94720, USA, (email: jiantao@eecs.berkeley.edu).}
\thanks{Lele Wang is with the Department of Electrical and Computer Engineering, University of British Columbia, Vancouver, BC V6T1Z4, Canada (email: lelewang@ece.ubc.ca).}
}

\begin{document}
\maketitle

\begin{abstract}
We consider the problem of computing a function of $n$ variables using noisy queries, where each query is incorrect with some fixed and known probability $p \in (0,1/2)$. Specifically, we consider the computation of the $\mathsf{OR}$ function of $n$ bits (where queries correspond to noisy readings of the bits) and the $\mathsf{MAX}$ function of $n$ real numbers (where queries correspond to noisy pairwise comparisons). We show that an expected number of queries of
\[
(1 \pm o(1)) \frac{n\log \frac{1}{\delta}}{D_{\mathsf{KL}}(p \| 1-p)}
\]
is both sufficient and necessary to compute both functions with a vanishing error probability $\delta = o(1)$, where $D_{\mathsf{KL}}(p \| 1-p)$ denotes the Kullback-Leibler divergence between $\mathsf{Bern}(p)$ and $\mathsf{Bern}(1-p)$ distributions. Compared to previous work, our results tighten the dependence on $p$ in both the upper and lower bounds for the two functions.
\end{abstract}

\section{Introduction}
Coping with noise in computing is an important problem to consider in large systems. With applications in fault tolerance~\citep{Hastad1987,Pease1980,pippenger1991lower}, active ranking~\citep{Shah2018,Agarwal2017,Falahatgar2017,Heckel2019,wang2022noisy,Gu-Xu2023}, noisy searching~\citep{Berlekamp1964,Horstein1963,Burnashev1974,pelc1989searching, Karp2007}, among many others, the goal is to devise algorithms that are robust enough to detect and correct the errors that can happen during the computation. More concretely, the problem can be defined as follows: suppose an agent is interested in computing a function $f$ of $n$ variables with an error probability at most $\delta$, as quickly as possible. To this end, the agent can ask binary questions (referred to hereafter as \emph{queries}) about the variables in hand. The binary response to the queries is observed by the agent through a binary symmetric channel (BSC) with crossover probability $p$. The agent can adaptively design subsequent queries based on the responses to previous queries. The goal is to characterize the relation between $n$, $\delta$, $p$, and the query complexity, which is defined as the minimal number of queries that are needed by the agent to meet the intended goal of computing the function $f$.

This paper considers the computation of the \orn function of $n$ bits and the \maxn function of $n$ real numbers. For the computation of the \orn function, the noisy queries correspond to noisy readings of the bits, where at each timestep, the agent queries one of the bits, and with probability $p$, the wrong value of the bit is returned. For the computation of the \maxn function, the noisy queries correspond to noisy pairwise comparisons, where at each timestep, the agent asks to compare two numbers in the given sequence, and with probability $p$, the incorrect result of the comparison is returned. 

This model for the \orn and \maxn functions has been previously considered in~\citep{feige1994computing}, where upper and lower bounds on the number of queries that are needed to compute the two functions are derived in terms of the number of variables $n$, the noise probability $p$ and the desired error probability $\delta$. The upper and lower bounds in~\citep{feige1994computing} are tight only in terms of dependence on $n$; however, the exact dependence of the query complexity on the parameters $p$ and $\delta$ is not very clear. In this paper, we derive new upper and lower bounds for the computation of the \orn and \maxn functions, with tight dependence on $p$ and $\delta$.

\subsection{Main Contribution}
The main result of this paper can be stated as follows: it is both sufficient and necessary to use
\[
    (1 \pm o(1)) \frac{n\log \frac{1}{\delta}}{D_{\mathsf{KL}}(p \| 1-p)}
\]
queries in expectation to compute the \orn and \maxn functions with a vanishing error probability $\delta = o(1)$, where $D_{\mathsf{KL}}(p \| 1-p)$ denotes the Kullback-Leibler divergence between $\mathsf{Bern}(p)$ and $\mathsf{Bern}(1-p)$ distributions. This result simultaneously improves the dependence on $p$ in both the upper and lower bounds derived in~\citep{feige1994computing}. 

From a technical standpoint, our lower bound for both functions is based on Le Cam's two point method, which transforms the noisy computing problem in hand into a proper binary testing problem between two instances of the input variables. By carefully designing the two instances to look at, Le Cam's two point lemma~\citep{yu1997assouad} then gives a non-asymptotic lower bound on the error probability of any estimator that wishes to distinguish the evaluation of the function over those two instances. 

The algorithms that attain the upper bounds are executed in two steps. In the first step, a subset of the input variables of sublinear size is chosen. In the second step, an existing algorithm (e.g., one of the algorithms in~\citep{feige1994computing}) is executed over the chosen subset. For the \orn function, the chosen subset corresponds to the input variables that, when queried, return a value of 1. For the \maxn function, the subset corresponds to the elements of the sequence that appear to be larger than the maximum of a random subsample of the input variables (of sublinear size). We show that $(1+o(1))\frac{n\log \frac{1}{\delta}}{D_{\mathsf{KL}}(p \| 1-p)}$ queries are used on average in the first step, while the query complexity of the second step ends up being a lower order term.

\subsection{Related Work}
The noisy computation of the \orn function has been first studied in the literature of circuit with noisy gates~\citep{dobrushin1977lower, dobrushin1977upper, von1956probabilistic, pippenger1991lower, gacs1994lower} and noisy decision trees~\citep{feige1994computing, evans1998average, reischuk1991reliable}. Since the querying model for the \orn function corresponds to noisy reading of the bits, it is closely related to the best arm identification problem, where the variables correspond to real-valued arms and the goal is to identify the arm with the largest value (i.e., reward)~\citep{bubeck2009pure, audibert2010best, garivier2016optimal, jamieson2014best, gabillon2012best, kaufmann2016complexity}. Indeed, any best-arm identification algorithm can be converted to an \orn computation since the evaluation of the \orn function of $n$ bits is equal to the evaluation of their maximum. This observation gives an upper bound $\co\left(\frac{n\log(1/\delta)}{1-H(p)}\right)$ on the number of queries needed for the computation of the \orn function~\citep{audibert2010best}, where $H(\cdot)$ denotes the binary entropy function. However, as we shall see in this paper, this bound doesn't give a tight dependence on $p$.

On the other hand, the noisy computation of the \maxn function has closely followed the literature on the active ranking problem (also known as ``noisy sorting'') and the top-$k$ selection problem~\citep{Mohajer2017,Shah2018,Agarwal2017,Falahatgar2017,Heckel2019,wang2022noisy, Wang2023,Gu-Xu2023}, where the ranking of a subset of the $n$ items is required. Many of these works considered generalized noise models, where the noise probability can depend on the pair of elements to be compared. For our simplified noise model, however, the best known upper and lower bounds for the computation of the \maxn function are due to~\citep{feige1994computing} which gives an upper bound based on the popular knockout tournament algorithm and a lower bound of $\Omega\left(\frac{n\log(1/\delta)}{\log((1-p)/p)}\right)$ queries. As we shall see, the dependence on $p$ in both bounds can be improved.

\subsection{Paper Organization and Notation}
This paper is organized as follows. In Section~\ref{sec:problem}, the noisy \orn and noisy \maxn problems are formally defined, and the main results are stated. In Section~\ref{sec:or}, we derive the upper and lower bounds for the noisy \orn problem, whereas in Section~\ref{sec:max}, we derive the upper and lower bounds for the noisy \maxn problem. In Section~\ref{sec:extension}, we discuss the extension of the model to the case of a varying crossover probability $p$ that depends on $n$.

Notation: The following notation will be used in this paper. Sets are denoted in calligraphic letters, random variables in upper-case letters, and realizations of random variables in lower-case letters, i.e., $\cx$ and $\cy$ are two sets, $X$ and $Y$ are random variables, and $x$ and $y$ are their respective realizations. Bold face letters are used for vectors, whereas normal face letters are used for scalars. The probability of an event $A$ is denoted as $\P(A)$, and the expectation of real-valued random variable $X$ is denoted as $\E[X]$. We will often use the shorthand notation $\E_x(Y)$ to denote the conditional expectation $\E[Y \cond X = x]$, where the random variable $X$ would be clear from the context. For a finite set $\cs$, we write $|\cs|$ to denote the cardinality of $\cs$. For $p \in (0,1)$, $\DKL \triangleq (1-2p)\log\left(\frac{1-p}{p}\right)$ denotes the Kullback-Leibler divergence between $\mathsf{Bern}(p)$ and $\mathsf{Bern}(1-p)$ distributions, and $H(p)\triangleq -p\log p-(1-p)\log (1-p)$ denotes the binary entropy function.

\section{Problem Formulation and Main Results} \label{sec:problem}
\subsection{Noisy \orn Problem}
Let us first consider the problem of computing the $\orn$ function of $n$ bits using noisy observations of the bits. To this end, let ${\bf X} = (X_1,\ldots,X_n) \in \{0,1\}^n$ be a sequence of $n$ bits. The \orn function is defined as
\begin{align}\label{eqn:def_or}
    \mathsf{OR}({\bf X}) = \begin{cases}1, & \text{ if } \exists \,i\in[n]: X_i = 1 \\ 
    0, & \text{ otherwise.}
    \end{cases}
\end{align}
At the $k$th time step, an agent can submit a query $U_k$ corresponding to one of the bits, i.e., $U_k \triangleq X_i$ for some $i \in [n]$. The response $Y_k$ can be expressed as
\[
Y_k = U_k \oplus Z_k,
\]
where $Z_k \sim \mathrm{Bern}(p)$ is a sequence of i.i.d. random variables. Note that the queries are designed adaptively, i.e.,
\[
U_k = f_k({\bf U}^{k-1},{\bf Y}^{k-1}),
\]
for some (possibly stochastic) function $f_k$. Moreover, the querying strategy may have a variable stopping time that depends on the query responses, i.e., the number of queries $M$ may be random. The goal of the agent is to estimate the \orn function using the noisy queries and responses. 

\subsection{Noisy \maxn Problem}
With a slight abuse of notation, for the noisy \maxn problem, we use ${\bf X} = (X_1,\ldots,X_n) \in \mathbb{R}^n$ as a sequence of $n$ real numbers. The \maxn function is defined as
\[
\maxn({\bf X})=\argmax_{i\in[n]}\, X_i.
\]
At the $k$th time step, the agent can submit a query $(U_k,V_k) \triangleq (X_i,X_j)$ for some $i,j \in [n]$. The response $Y_k$ can be expressed as
\[
Y_k = \mathbbm{1}_{\{U_k <V_k\}} \oplus Z_k,
\]
where $\mathbbm{1}_{\{.\}}$ denotes the indicator function, and $Z_k \sim \mathrm{Bern}(p)$ is a sequence of i.i.d. random variables. As for the \orn problem, the queries can be designed adaptively, i.e.,
\[
(U_k,V_k) = g_k({\bf U}^{k-1},{\bf V}^{k-1},{\bf Y}^{k-1}),
\]
for some (possibly stochastic) function $g_k$. Moreover, the querying strategy may have a variable stopping time. The goal of the agent is to estimate the \maxn function using the responses to the noisy pairwise comparisons.

\begin{remark} \label{remark:or_max_relation}
    Notice the difference in the querying models for the noisy \orn and \maxn problems. While the queries in the noisy \orn problem correspond to readings of the bits, the queries in the noisy \maxn problem correspond to pairwise comparisons between the elements in the given sequence. Due to this distinction, one cannot immediately translate any upper or lower bound on the number of queries from one problem to the other (despite the fact that, for example, the \orn of $n$ bits is equal to their maximum).
\end{remark}

\subsection{Main Results}
\label{sec:main-results}
Our main results can be summarized in the following four theorems, which give tight upper and lower bounds on the expected number of queries needed to compute the noisy \orn and noisy \maxn functions.

\begin{theorem}[Lower bound for $\mathsf{OR}$] \label{thm:or-lower} 
    For the noisy \orn problem, it holds that
    \begin{align*}
        \inf_{\hat\mu} \sup_{{\bf x}\in\{0, 1\}^n} \P(\hat\mu \neq \mathsf{OR}({\bf x}))\geq \frac{1}{4}\cdot \exp\left(-\frac{m\cdot D_{\mathsf{KL}}(p\|1-p)}{n }\right),
    \end{align*}
    where $m$ is the expected number of queries under the worst-case input binary sequence.
    % i.e., 
    % \[
    %     m = \sup_{\bf x}\E[M \cond {\bf X}= {\bf x}].
    % \]
    Thus, the expected number of queries needed to compute the \orn function with worst-case error probability at most $\delta = o(1)$ is lower bounded by $(1 - o(1)) \frac{n\log \frac{1}{\delta}}{D_{\mathsf{KL}}(p \| 1-p)}$.
\end{theorem}

\begin{theorem}[Upper bound for $\mathsf{OR}$] \label{thm:or-upper}
    For any vanishing $\delta=o(1)$, there exists an algorithm (Algorithm~\ref{alg:noisyor} in Section~\ref{sec:or-upper}) that computes the \orn function with an expected number of queries at most
    \begin{equation}
        (1+o(1))\frac{n\log \frac{1}{\delta}}{\DKL}
    \end{equation}
    and worst-case error probability at most $\delta$.
\end{theorem}

\begin{theorem}[Lower bound for $\mathsf{MAX}$]\label{thm:max-lower} 
    For the noisy \maxn problem, it holds that
    \begin{align*}
        \inf_{\hat\mu} \sup_{{\bf x}\in \mathbb{R}^n} \P(\hat\mu \neq \maxn({\bf x}))\geq \frac{1}{4}\cdot \exp\left(-\frac{m \cdot D_{\mathsf{KL}}(p\|1-p)}{n }\right),
    \end{align*}
    where $m$ is the expected number of queries under the worst-case input sequence. Thus, the expected number of queries needed to compute the \maxn function with worst-case error probability at most $\delta = o(1)$ is lower bounded by $(1 - o(1)) \frac{n\log \frac{1}{\delta}}{D_{\mathsf{KL}}(p \| 1-p)}$.
\end{theorem}

\begin{theorem}[Upper bound for $\mathsf{MAX}$] \label{thm:max-upper}
    For any vanishing $\delta=o(1)$, there exists an algorithm (Algorithm~\ref{alg:noisymax} in Section~\ref{sec:max-upper}) that computes the $\maxn$ function with an expected number of queries at most
    \[
    (1+o(1))\frac{n\log \frac{1}{\delta}}{D_\mathsf{KL}(p||1-p)}
    \]
    and worst-case error probability at most $\delta$.
\end{theorem}

\section{Noisy \orn Problem} \label{sec:or}

\subsection{Lower Bound} \label{sec:or-lower}
First, we consider the lower bound for the noisy \orn problem (Theorem~\ref{thm:or-lower}). We give a proof of the theorem based on Le Cam's two point method \citep{lecam1973convergence, yu1997assouad}, which is stated in  Lemma~\ref{lem:lecam} in Appendix~\ref{app:lemma} for the reader's convenience.

\begin{proof}[Proof of Theorem~\ref{thm:or-lower}]
Let ${\bf x}_0$ be the length-$n$ all-zero sequence, and let ${\bf x}_j \in \{0,1\}^n$ be such that $x_{jj} = 1$ and $x_{ji} = 0$ for $i \neq j$, where $x_{ji}$ refers to the $i$-th element in the $j$th binary sequence ${\bf x}_j$. We can first verify that for any estimate $\hat\mu$, one has
\begin{align*}
    \mathds{1}(\hat\mu \neq \mathsf{OR}({\bf x}_0)) +  \mathds{1}(\hat\mu \neq \mathsf{OR}({\bf x}_j)) \geq 1.
\end{align*}
By applying Le Cam's two point lemma on ${\bf x}_0$ and ${\bf x}_j$, we know that
\begin{align*}
      \inf_{\hat\mu} \sup_{{\bf x}\in\{0, 1\}^n} \P(\hat\mu \neq \mathsf{OR}({\bf x})) & \geq \frac{1}{2}\left(1-\mathsf{TV}(\mathbb{P}_{{\bf x}_0}, \mathbb{P}_{{\bf x}_j})\right). 
\end{align*}
Here, $\mathbb{P}_{{\bf x}_j}$ denotes the distribution of the query responses in $m$ rounds when the underlying binary sequence is ${\bf x}_j$. By taking maximum over $j$ on the right-hand side, we have
\begin{align*}
     \inf_{\hat\mu} \sup_{{\bf x}\in\{0, 1\}^n} \P(\hat\mu \neq \mathsf{OR}({\bf x})) & \geq \sup_{1\leq j\leq n}\frac{1}{2}\left(1-\mathsf{TV}(\mathbb{P}_{{\bf x}_0}, \mathbb{P}_{{\bf x}_j})\right)  \\ 
     & \stackrel{(a)}{\geq} \sup_{1\leq j\leq n} \frac{1}{4}\exp\left(-D_{\mathsf{KL}}(\mathbb{P}_{{\bf x}_0}, \mathbb{P}_{{\bf x}_j})\right),
\end{align*}
where $(a)$ follows due to the Bretagnolle–Huber inequality~\citep{bretagnolle79} (Lemma~\ref{lem:bh} in Appendix~\ref{app:lemma}). Now, we aim at computing $D_{\mathsf{KL}}(\mathbb{P}_{{\bf x}_0}, \mathbb{P}_{{\bf x}_j})$. Let $M_j$ be the random variable that denotes the number of times the $j$th element is queried among all $m$ rounds.  From the divergence decomposition lemma~\citep{auer1995gambling} (Lemma~\ref{lem:div} in Appendix~\ref{app:lemma}) and the fact that ${\bf x}_0$ and ${\bf x}_j$ only differ in the $j$th position, we have  
\begin{align*}
   D_{\mathsf{KL}}(\mathbb{P}_{{\bf x}_0}, \mathbb{P}_{{\bf x}_j}) = \E_{{\bf x}_0}[M_j] \cdot D_{\mathsf{KL}}(p\|1-p).
\end{align*}
Here, $\mathbb{E}_{{\bf x}_0}[M_j]$ denotes the expected number of times the $j$th element is queried when the underlying sequence is ${\bf x}_0$. Thus we have 
\begin{align*}
     \inf_{\hat\mu} \sup_{{\bf x}\in\{0, 1\}^n} \P(\hat\mu \neq \mathsf{OR}({\bf x})) \geq \sup_{1\leq j\leq n} \frac{1}{4}\exp\left(-\E_{{\bf x}_0}[M_j] \cdot D_{\mathsf{KL}}(p\|1-p)\right).
\end{align*}
Now, since $\sum_j \E_{{\bf x}_0}[M_j] = \E_{{\bf x}_0}[M] \leq \sup_{\bf x}\E_{\bf x}[M] \triangleq m$, there must exist some $j$ such that $\E_{{\bf x}_0}[M_j]  \leq m/n$. This gives
\begin{align*}
      \inf_{\hat\mu} \sup_{{\bf x}\in\{0, 1\}^n} \P(\hat\mu \neq \mathsf{OR}({\bf x})) &  \geq  \frac{1}{4} \exp\left(-\frac{m\cdot D_{\mathsf{KL}}(p\|1-p)}{n}\right).
\end{align*}
Therefore, to achieve an error probability $\delta$ for any input sequence ${\bf x}$, the number of queries must satisfy that
\[
m\geq \frac{n\log \frac{1}{4\delta}}{D_\mathsf{KL}(p||1-p)}=(1-o(1))\frac{n\log \frac{1}{\delta}}{D_\mathsf{KL}(p||1-p)}.
\]
\end{proof}

\subsection{Upper Bound} \label{sec:or-upper}
Now, we prove the upper bound for the noisy \orn problem (Theorem~\ref{thm:or-upper}). Towards this end, we first describe two existing algorithms that will be used as subroutines in our proposed algorithm for the noisy \orn problem. The first algorithm is the \scb algorithm, which takes as input a single bit and returns an estimate of the bit through repeated queries. The following lemma holds for the \scb algorithm, which is given for completion in Algorithm~\ref{alg:checkbit}. Note that the parameter $\alpha$ in the $i$th iteration of Algorithm~\ref{alg:checkbit} corresponds to the probability that the input bit is one given the first $i$ noisy observations.

\begin{algorithm}[t]%[!htbp]
    \caption{\scb{} algorithm for finding the value of an input bit}
    \label{alg:checkbit}
    \textbf{Input}: Input bit $x$, tolerated error probability $\delta$, crossover probability $p$.\\
    \textbf{Output}: Estimate of $x$.
    \begin{algorithmic}[1]
        \State Set $\alpha\gets 1/2$.
        \While{$\alpha\in(\delta, 1-\delta)$}
            \State Query the input bit $x$.
            \If{$1$ is observed}
                \State Set $\alpha \gets \frac{(1-p)\alpha}{(1-p)\alpha + p(1-\alpha)}$.
            \Else
                \State Set  $\alpha \gets \frac{p\alpha}{p\alpha + (1-p)(1-\alpha)}$. 
            \EndIf
        \EndWhile
        \If{$\alpha\ge 1-\delta$}
            \State \Return $1$.
        \Else
            \State \Return $0$.
        \EndIf
    \end{algorithmic}
\end{algorithm}

\begin{lemma}[Lemma 13 in~\citep{Gu-Xu2023}] \label{lemma:check-bit}
    There exists a randomized algorithm, namely, the \scb{} algorithm, that finds the value of any input bit $x$ with error probability at most $\delta$. The algorithm makes at most 
       $ \frac{1}{1-2p}\left\lceil\frac{\log\frac{1-\delta}{\delta}}{\log\frac{1-p}{p}}\right\rceil$
    queries in expectation.
\end{lemma}

The second algorithm is the \stouror algorithm for computing the \orn function from noisy queries. This algorithm has been previously proposed in our earlier work~\citep{zhu2023optimal}, and is given here for completion in Algorithm~\ref{alg:tournament_or}. The following lemma gives an upper bound on the expected number of queries made by the \stouror algorithm.

\begin{lemma}[Theorem J.1 in~\citep{zhu2023optimal}] \label{lemma:tournament-or}
    % Suppose the crossover probability $p$ is a constant bounded away from $0$ and $1/2$. 
    For any $\delta<1/2$, there exists an algorithm, namely the \stouror algorithm, that computes the \orn function with error probability at most $\delta$. The expected number of queries made by the algorithm is upper bounded by $C\left(\frac{n}{1-H(p)}+\frac{n\log(1/\delta)}{\DKL}\right)$, where $C$ is an absolute constant.
\end{lemma}

\begin{algorithm}[t]%[!htbp]
    \caption{\stouror algorithm for computing noisy \orn function}
    \label{alg:tournament_or}
    \textbf{Input}: Bit sequence ${\bf x} = (x_1,\ldots,x_n)$, tolerated error probability $\delta$, crossover probability $p$.\\
    \textbf{Output}: Estimate of $\orn({\bf x})$.
    \begin{algorithmic}[1]
        \State Set ${\bf y} \gets {\bf x}$, $r\gets n$.
        \For{iteration $i= 1:\lceil\log_2(n)\rceil $}
            \State Set $\tilde\delta_i \gets \delta^{2(2i-1)}$.
            \For{iteration $j = 1:\lfloor r/2\rfloor$}
                \State Set $a \gets \scb(y_{2j-1}, \, \tilde\delta_i, \, p)$.
                \If{$a = 1$}
                    \State Set $z_j \gets y_{2j-1}$.
                \Else
                    \State Set $z_j \gets y_{2j}$.
                \EndIf
            \EndFor
            \If{$r$ is even}
                \State Set ${\bf y} \gets (z_1, \ldots, z_{r/2})$.
            \Else
                \State Set ${\bf y} \gets (z_1, \ldots, z_{\lfloor r/2\rfloor}, y_r)$.
            \EndIf
            \State Set $r\gets \lceil r/2 \rceil$.
        \EndFor
        \State \Return \scb($y$, $\delta$, p)
    \end{algorithmic}
\end{algorithm}

Now, we are ready to describe the proposed \sno algorithm for the noisy \orn problem. The algorithm is given in Algorithm~\ref{alg:noisyor}, in which the \scb and \stouror algorithms appear as subroutines. In the following, we analyze the error probability and the expected number of queries made by the \sno algorithm.

\begin{algorithm}[t]%[!htbp]
    \caption{Proposed \sno algorithm}
    \label{alg:noisyor}
    \textbf{Input}: Bit sequence ${\bf x} = (x_1,\ldots,x_n)$, tolerated error probability $\delta$, crossover probability $p$.\\
    \textbf{Output}: Estimate of $\orn({\bf x})$.
    \begin{algorithmic}[1]
        \State Set ${\bf y} \gets \emptyset$, $r \gets 0$.
        \For{$i\in [n]$}
            \If{\scb($x_i$, $\delta$, $p)=1$}
                \State Append $x_i$ to ${\bf y}$.
                \State Set $r \gets r + 1$.
            \EndIf
        \EndFor
        \If{$r=0$}
            \State \Return $0$
        \ElsIf{$r\ge \max(\log n, \, n\delta\log\frac{1}{\delta})$}
            \State \Return $1$
        \Else
            \State \Return \stouror(${\bf y}$, $\delta$, $p$)
        \EndIf
    \end{algorithmic}
\end{algorithm}

\vspace{0.5em}
\subsubsection{Error Analysis}
For the \sno algorithm, the following lemma holds.
\begin{lemma} \label{lemma:sno_error}
    For any input parameter $\delta < 1/2$ and bit sequence ${\bf x}$, the error probability of the \sno algorithm is upper bounded by $2\delta$.
\end{lemma}

\begin{proof}
    Let $\hat\mu$ denote the output of the \sno algorithm, and let $\cs$ denote the set of bits appended to ${\bf y}$ in Line 4 of the algorithm. To prove Lemma~\ref{lemma:sno_error}, we are going to separately consider two types of input sequences ${\bf x}$. 

    First, suppose that at least one bit of the input sequence is $1$, i.e., $\orn({\bf x})=1$. Without loss of generality, assume that $x_1=1$. Let $\mathcal{E}_1\triangleq\{\hat\mu=0\}$ be the error event in this case, and let $\mathcal{A}\triangleq\{x_1\notin \mathcal{S}\}$ denote an auxiliary event that $x_1$ is not appended into ${\bf y}$. By the union bound, we have 
    \begin{equation*}
        \P(\mathcal{E}_1) \le \P(\ca)+\P(\ce_1 \cond \ca^c).
    \end{equation*}
    For the first term, we notice that $\P(\ca)\le \delta$ by Lemma~\ref{lemma:check-bit}. For the second term, note that $|\mathcal{S}|\ge 1$ on the event $\ca^c$. Therefore, the only possibility for the algorithm to return $0$ is that the function call of \stouror on Line 11 incorrectly returns $0$. By Lemma~\ref{lemma:tournament-or}, we have that $\P(\ce_1 \cond \ca^c)\le \delta$. This leads to the upper bound $\P(\ce_1)\le 2\delta$.

    Next, we move on to consider the all-zero input sequence. Let $\ce_0\triangleq\{\hat\mu=1\}$ denote the error event in this case, and define an auxiliary event $\cb=\{|\mathcal{S}|\ge \max(\log n, \, n\delta\log\frac{1}{\delta})\}$. As before, we have
    \begin{equation*}
        \P(\mathcal{E}_0)\le \P(\cb)+\P(\ce_0|\cb^c).
    \end{equation*}
    By Lemma~\ref{lemma:tournament-or}, we know that $\P(\ce_0|\cb^c)\le \delta$. Moreover, we claim that $\P(\cb)=o(\delta)$. These would together imply that $\P(\ce_0)=(1+o(1))\delta$.

    Now we move on to show the claim. By Lemma~\ref{lemma:check-bit}, we know that the random variable $|\mathcal{S}|$ is statistically dominated by $\mathrm{Binom}(n,\delta)$. We will separately consider two different regimes to show the claim.

    \vspace{0.25em}
    \noindent \textbf{Regime 1: $\delta\ge \frac{1}{n}$.} In this case, $\max(\log n, \,n\delta\log\frac{1}{\delta})=n\delta\log\frac{1}{\delta}$. We have
    \begin{align*}
        \P\left(|\mathcal{S}|\ge n\delta\log\frac{1}{\delta}\right)&\le \P\left(\mathrm{Binom}(n,\delta)\ge n\delta\log\frac{1}{\delta}\right)\\
        &\stackrel{(a)}{\le}\exp\left(-n D_{\mathsf{KL}}\left(\delta\log\frac{1}{\delta}\Big\|\delta\right)\right)\\
        &=\exp\left(-n \delta\log\frac{1}{\delta}\log\log\frac{1}{\delta}(1-o(1))\right)\\
        &\stackrel{(b)}{=}o(\delta),
    \end{align*}
    where $(a)$ follows by the additive Chernoff bound (Lemma~\ref{lemma:chernoff} in Appendix~\ref{app:lemma}), and $(b)$ follows because $n\delta\ge 1$ and $\log\log \frac{1}{\delta}=\omega(1)$.

    \vspace{0.25em}
    \noindent \textbf{Regime 2: $\delta<\frac{1}{n}$.} In this case, $\max(\log n, \, n\delta\log\frac{1}{\delta})=\log n$. We have
    \begin{align*}
        \P(|\mathcal{S}|\ge \log n)&\le \P\left(\mathrm{Binom}(n,\delta)\ge \log n\right)\\
        &\le \exp\left(-n D_{\mathsf{KL}}\left(\frac{\log n}{n}\Big\|\delta\right)\right)\\
        &\stackrel{(c)}{=}\exp\left(-\log n\log\frac{\log n}{n\delta}(1-o(1))\right)\\
        &=\left(\frac{n\delta}{\log n}\right)^{(1-o(1))\log n},
    \end{align*}
    where $(c)$ follows because $\delta<1/n$. To complete the argument, first suppose that $\delta=n^{-\Theta(1)}$. In this case, because $n\delta<1 $, we have $$\left(\frac{n\delta}{\log n}\right)^{(1-o(1))\log n}
    <\frac{1}{(\log n)^{(1-o(1))\log n}}=n^{-\omega(1)}=o(\delta).$$ Secondly, suppose $\delta=n^{-\omega(1)}$. Then we have
    \[\exp\left(-\log n\log\frac{\log n}{n\delta}(1-o(1))\right)=\exp\left(-\log n\log \frac{1}{\delta}(1-o(1))\right)=o(\delta).\]
\end{proof}

\vspace{0.5em}
\subsubsection{Query Analysis}
Let $M$ denote the number of queries made by the \sno algorithm. The following lemma provides an upper bound on the expectation of $M$, and hence completes the proof of Theorem~\ref{thm:or-upper}.
\begin{lemma} \label{lemma:number-queries-or}
    % Suppose $\delta=o(1)$ and $p$ is a constant bounded away from $0$ and $1/2$. 
    For $\delta = o(1)$, the expected number of queries made by the \sno algorithm satisfies that
    \[
    \E[M]\le(1+o(1))\frac{n\log\frac{1}{\delta}}{\DKL}.
    \]
\end{lemma}
\begin{proof}
    As before, let $\cs$ denote the set of bits appended to ${\bf y}$ in Line 4 of the algorithm. To bound $\E[M]$, we only need to bound the expected number of queries made by the $n$ calls of the \scb function on Line 3, and the single call of the \stouror function on Line 11.
    By Lemma~\ref{lemma:check-bit}, we know that the expected number of queries spent at Line 3 is at most
    \[
    \frac{n}{1-2p}\left\lceil\frac{\log\frac{1-\delta}{\delta}}{\log\frac{1-p}{p}}\right\rceil\le \frac{n}{1-2p}\left(\frac{\log\frac{1-\delta}{\delta}}{\log\frac{1-p}{p}}+1\right)=(1+o(1))\frac{n\log\frac{1}{\delta}}{(1-2p)\log\frac{1-p}{p}}=(1+o(1))\frac{n\log \frac{1}{\delta}}{\DKL},
    \]
    where the penultimate equality follows because $\delta=o(1)$ and $p$ is bounded away from $0$ and $1/2$, and the last equality follows because $\DKL\triangleq (1-2p)\log\frac{1-p}{p}$.
    By Lemma~\ref{lemma:tournament-or}, we know that the expected number of queries spent at Line 11 is at most
    \[
        C\left(\frac{\E[|\mathcal{S}|]}{1-H(p)}+\frac{\E[|\mathcal{S}|]\log(1/\delta)}{\DKL}\right)=C(1+o(1))\E[|\mathcal{S}|]\log(1/\delta)=o(n\log(1/\delta)),
    \]
    where the penultimate equality follows because $\delta=o(1)$ and $p$ is bounded away from $0$ and $1/2$, and the last equality follows because we always have $|\mathcal{S}|=o(n)$ on the event that Line 11 is executed.

    Finally, by the linearity of expectation, we have $\E[M]\le (1+o(1))\frac{n\log \frac{1}{\delta}}{\DKL}$.
\end{proof}

\section{Noisy \maxn Problem} \label{sec:max}

\subsection{Lower Bound} \label{sec:max-lower}
In this part, we turn our attention to the noisy \maxn problem. We first consider the lower bound established in Theorem~\ref{thm:max-lower}. As for the \orn problem, our proof of the lower bound is based on Le Cam's two point method (see Lemma~\ref{lem:lecam} in Appendix~\ref{app:lemma}).

\begin{proof}[Proof of Theorem~\ref{thm:max-lower}]
    Consider an arbitrary sequence ${\bf y} \in \mathbb{R}^n$. Without loss of generality, assume that $y_1 < \cdots < y_n$. Clearly, in this case, we have $\maxn({\bf y}) = n$. Now, for $i\in[n-1]$, consider the sequence ${\bf z}_i = (y_{1},\ldots, y_{i-1}, y_{n}, y_{i}, \ldots, y_{n-1})$, i.e., we move the $n$th element in ${\bf y}$ and insert it between $y_{i-1}$ and $y_i$. Notice that $\maxn({\bf z}_i)=i$ for all $i\in[n-1]$. 
    
    Let $M$ denote the total number of queries made by the estimator, and let $M_{i,j}$ represent the number of pairwise comparisons made between the $i$th and $j$th elements in the $M$ rounds. We first notice that, for any estimator $\hat\mu$ and any $i \in [n-1]$,
    \begin{align*}
        \mathds{1}(\hat\mu \neq \maxn({\bf y})) +  \mathds{1}(\hat\mu \neq \maxn({\bf z}_i)) \geq 1.
    \end{align*}
    Following a similar proof as Theorem~\ref{thm:or-lower}, it holds that
    \begin{align*}
          \inf_{\hat\mu} \sup_{{\bf x}\in \mathbb{R}^n} \P(\hat\mu \neq \mathsf{MAX}({\bf x}))
          & \stackrel{(a)}{\geq} \sup_{1\leq i\leq n-1}\frac{1}{2}\left(1-\mathsf{TV}(\mathbb{P}_{{\bf y}}, \mathbb{P}_{{\bf z}_i})\right)  \\ 
          & \stackrel{(b)}{\geq} \sup_{1\leq i\leq n-1} \frac{1}{4}\exp\left(-D_{\mathsf{KL}}(\mathbb{P}_{{\bf y}}, \mathbb{P}_{{\bf z}_i})\right) \\
          & \stackrel{(c)}{\geq} \sup_{1\leq i\leq n-1} \frac{1}{4}\exp\left(-\sum_{j=i+1}^n \E_{{\bf y}}[M_{i,j}] \cdot D_{\mathsf{KL}}(p\|1-p)\right),
    \end{align*}
    where $\mathbb{P}_{{\bf z}_i}$ denotes the distribution of query responses in $m$ rounds when the underlying sequence is ${\bf z}_i$. Here, $(a)$ follows by Le Cam's two point lemma (Lemma~\ref{lem:lecam} in Appendix~\ref{app:lemma}), $(b)$ follows from the Bretagnolle–Huber inequality (Lemma~\ref{lem:bh} in Appendix~\ref{app:lemma}), and $(c)$ follows from the divergence decomposition lemma (Lemma~\ref{lem:div} in Appendix~\ref{app:lemma}) and the fact that the distributions $\mathbb{P}_{{\bf y}}$ and $\mathbb{P}_{{\bf z}_i}$ only differ when the $i$th element is compared with a subsequent element. Now, since $\sum_{i,j\in[n], i< j} \E_{{\bf y}}[M_{i,j}] = \E_{{\bf y}}[M] \leq \sup_{\bf x}\E_{\bf x}[M] \triangleq m$, there must exist some index $i$ such that $\sum_{j=i+1}^n \E_{{\bf y}}[M_{i,j}]  \leq \frac{m}{n-1}$. This gives
    \begin{align*}
          \inf_{\hat\mu} \sup_{{\bf x}\in\mathbb{R}^n} \P(\hat\mu \neq \maxn({\bf x})) &  \geq  \frac{1}{4}\exp\left(-\frac{m\cdot \DKL}{n-1}\right).
    \end{align*}
    It follows that, to achieve an error probability at most $\delta=o(1)$ for any input sequence ${\bf x}$, the number of queries must satisfy that
    \[
        m\geq \frac{(n-1)\log \frac{1}{4\delta}}{\DKL}=(1-o(1))\frac{n\log \frac{1}{\delta}}{\DKL}.
    \]
\end{proof}

\subsection{Upper Bound} \label{sec:max-upper}
Next, we prove the upper bound stated in Theorem~\ref{thm:max-upper}. To prove the theorem, we devise a querying strategy for the noisy \maxn problem. Towards this end, we first highlight the \textsc{NoisyCompare} algorithm for comparing two elements in the presence of noise with an error probability at most $\delta$. This function is given in Algorithm~\ref{alg:noisycompare}, which can be seen as a recast of Algorithm 1 in~\citep{Gu-Xu2023}. 

\begin{lemma}[Lemma 13 in~\citep{Gu-Xu2023}] \label{lemma:noisy-compare}
    The \textsc{NoisyCompare} algorithm has an error probability at most $\delta$, and the expected number of queries made by the algorithm is at most 
       $ \frac{1}{1-2p}\left\lceil\frac{\log\frac{1-\delta}{\delta}}{\log\frac{1-p}{p}}\right\rceil$.
\end{lemma}

\begin{algorithm}[t]%[htbp]
\caption{\textsc{NoisyCompare} for comparing two elements in the presence of noise} \label{alg:noisycompare}
\textbf{Input}: Elements $x_1$ and $x_2$ to be compared, tolerated error probability $\delta$, crossover probability $p$.\\
\textbf{Output}: Boolean variable indicating if $x_1$ is less than $x_2$.
\begin{algorithmic}[1]
    \State Set $\alpha \gets \frac12$.
    \While{$\mathrm{true}$}
        \State Make a noisy pairwise comparison between $x_1$ and $x_2$.
        \If{$x_1<x_2$ is returned from the comparison}
            \State Set $\alpha\gets \frac{(1-p)\alpha}{(1-p)\alpha+p(1-\alpha))}$.
        \Else
            \State Set $\alpha\gets \frac{p\alpha}{p\alpha+(1-p)(1-\alpha)}$.
        \EndIf
        \If{$\alpha\ge 1-\delta$}
        \State \Return $\mathrm{true}$
        \EndIf
        \If{$\alpha\le \delta$}
        \State \Return $\mathrm{false}$
        \EndIf
    \EndWhile
\end{algorithmic}
\end{algorithm}

\begin{remark}
    Note that to compute the noisy \maxn function, one can repeatedly apply the \textsc{NoisyCompare} algorithm $n-1$ times, while setting the tolerated error probability to $\frac{\delta}{n-1}$ each time. This would guarantee that the noisy \maxn function can be computed with an overall error probability $\delta$; however, the total expected number of queries in this case would be $(1 + o(1)) \frac{n\log \frac{n}{\delta}}{D_{\mathsf{KL}}(p \| 1-p)}$, which is not tight compared to the lower bound derived in Section~\ref{sec:max-lower}.
\end{remark}

Another subroutine that will be used in our proposed algorithm is the \stourmax algorithm, given in Algorithm~\ref{alg:tournament_max}, which can be used to compute the noisy \maxn function. The following lemma states that $O(n\log(1/\delta))$ queries are made by the algorithm on average. Note that the exact multiplicative constant may depend on $p$; however, for the sake of the analysis of our proposed algorithm, the exact constant does not matter.

\begin{lemma}[Theorem J.1 in~\citep{zhu2023optimal}] \label{lemma:tournament}
    The \stourmax algorithm computes the $\maxn$ function with error probability at most $\delta$. The algorithm makes at most $C\left(\frac{n}{1-H(p)}+\frac{n\log(1/\delta)}{\DKL}\right)$ queries on average, where $C$ is an absolute constant.
\end{lemma}

\begin{algorithm}[t]%[!htbp]
\caption{\stourmax algorithm for computing the noisy \maxn function}
\label{alg:tournament_max}
\textbf{Input}: Sequence ${\bf x} = (x_1,\ldots,x_n)$, tolerated error probability $\delta$, crossover probability $p$.\\
\textbf{Output}: Estimate of $x_{\maxn({\bf x})}$.
\begin{algorithmic}[1]
\State Set ${\bf y} \gets {\bf x}$, $r \gets n$.
\For{iteration $i= 1:\lceil\log_2(n)\rceil $}
    \State Set $\tilde\delta_i \gets \delta^{2(2i-1)}$.
    \For{iteration $j = 1:\lfloor r/2\rfloor$}
        \If{\textsc{NoisyCompare}($y_{2j-1}$, $y_{2j}$, $\tilde\delta_i$, $p$)}
            \State Set $z_j \gets y_{2j}$.
        \Else
            \State Set $z_j \gets y_{2j-1}$.
        \EndIf
    \EndFor
    \If{$r$ is even}
        \State Set ${\bf y} \gets (z_1, \ldots, z_{r/2})$.
    \Else
        \State Set ${\bf y} \gets (z_1, \ldots, z_{\lfloor r/2\rfloor}, y_r)$.
    \EndIf
    \State Set $r\gets \lceil r/2 \rceil$.
\EndFor
\State \Return $y$
\end{algorithmic}
\end{algorithm}

\begin{algorithm}[t]%[htbp]
    \caption{Proposed \textsc{NoisyMax} algorithm} \label{alg:noisymax}
    \textbf{Input}: Sequence ${\bf x} = (x_1,\ldots,x_n)$, tolerated error probability $\delta$, crossover probability $p$.\\
    \textbf{Output}: Estimate of $x_{\maxn({\bf x})}$.
    \begin{algorithmic}[1]
        \State Set ${\bf y} \gets \emptyset$, ${\bf z} \gets \emptyset$.
        \For{iteration $i= 1:n$}
            \State Append $x_i$ to ${\bf y}$ with probability $\frac{1}{\log n}$.
        \EndFor
        \State Let $\mathcal{S}$ be the set of elements of ${\bf x}$ that were appended to ${\bf y}$.
        \State Set $\hat{y}\gets$ \stourmax$({\bf y},\delta,p)$.
        \State Append $\hat{y}$ to ${\bf z}$.
        \For{element $u \in \mathcal{S}^c$}
            \If{\textsc{NoisyCompare}($\hat{y}$, $u$, $\delta$, $p$)}
                \State Append $u$ to ${\bf z}$.
            \EndIf
        \EndFor
        \State Set $\hat{x}\gets$ \stourmax$({\bf z},\delta,p)$.
        \State \Return $\hat{x}$
    \end{algorithmic}
\end{algorithm}

The proposed \textsc{NoisyMax} algorithm is given in Algorithm~\ref{alg:noisymax}, where the \stourmax and \textsc{NoisyCompare} algorithms appear as subroutines. In the following, we analyze the error probability and the number of queries of the proposed algorithm.

\vspace{0.5em}
\subsubsection{Error Analysis}
Let $x^{*}= x_{\maxn({\bf x})}$ denote the largest element in ${\bf x}$. In the following, we provide an upper bound on the error probability $\P(\hat{x}\neq x^*)$.

\begin{lemma}[Error probability of \textsc{NoisyMax}] \label{lemma:error-prob-max}
    The error probability of the \textsc{NoisyMax} algorithm satisfies that 
    \[
        \P(\hat{x}\neq x^*)\le 3\delta.
    \]
\end{lemma}

\begin{proof}
To prove Lemma~\ref{lemma:error-prob-max}, we first define a few error events. Let $\ce=\{\hat{x}\neq x^*\}$ denote the error probability of the \textsc{NoisyMax} algorithm. Let $y^{*} \triangleq \maxn({\bf y})$ and $z^{*} \triangleq \maxn({\bf z})$ denote the indices of the largest elements of ${\bf y}$ and ${\bf z}$, where ${\bf y}$ and ${\bf z}$ are as constructed in Lines 3 and 9 of Algorithm~\ref{alg:noisymax} respectively. We define $\ce_1=\{\hat{y}\neq y^{*}\}$ as the event that the \stourmax algorithm outputs an incorrect estimate in Line 5. We let $\ce_2=\{x^* \text{ is not contained in } {\bf z}\}$ denote the event that the true maximum element $x^*$ does not belong to the vector ${\bf z}$. Finally, we define $\ce_3 = \{\hat{x} \neq z^{*}\}$ as the event that the \stourmax algorithm in Line 10 returns an incorrect estimate. 

By the union bound, we clearly have
\[
\P(\ce) \leq \P(\ce_1) + \P(\ce_2 \cond \ce_1^c) + \P(\ce_3).
\]
By Lemma~\ref{lemma:tournament}, we have that $\P(\ce_1) \leq \delta$ and $\P(\ce_3) \leq \delta$. To bound $\P(\ce_2 \cond \ce_1^c)$, we consider two cases. First, in the case that $x^*$ is appended to ${\bf y}$ in Line 3, we have that $x^* = y^*$. Given the event $\ce_1^c$, we thus know that $\hat{y} = x^*$, which is appended to ${\bf z}$ in Line 6. Hence, we have that $\P(\ce_2 \cond \ce_1^c) = 0$ in this case. In the second case that $x^*$ is not appended to ${\bf y}$ in Line 3, we have that $x^* \in \cs^c$. It follows that $x^*$ is not appended to ${\bf z}$ only if the result of the noisy comparison of $\hat{y}$ and $x^*$ in Line 8 is incorrect. Since this happens with error probability at most $\delta$, we have that $\P(\ce_2 \cond \ce_1^c) \leq \delta$ in this case, and the statement of the lemma follows.
\end{proof}

\vspace{0.5em}
\subsubsection{Query Analysis}
Let $M$ denote the number of queries made by the \textsc{NoisyMax} algorithm. In the following lemma, we provide an upper bound on $\E[M]$.

\begin{lemma} \label{lemma:number-queries-max}
    Suppose $\delta=o(1)$. The expected number of queries $\E[M]$ of the \textsc{NoisyMax} algorithm satisfies that
    \[
        \E[M]\leq (1+o(1))\frac{n\log(1/\delta)}{D_\mathrm{KL}(p||1-p)}.
    \]
\end{lemma}

\begin{proof}
    To prove Lemma~\ref{lemma:number-queries-max}, we first define two auxiliary sets. Let $\cs$ be the set of elements $x_i$ that were appended to ${\bf y}$ (as defined in Line 4 of Algorithm~\ref{alg:noisymax}), and let $\ct$ be the set of elements $x_i$ that were appended to ${\bf z}$ in Lines 6 and 9. Consider the two events:
    \begin{align*}
        \ca &= \{|\cs|\le \frac{n}{\log n}+n^{3/4}\},\\
        \cb &= \{|\ct|\le \sqrt{n}+\delta n+n^{3/4}\}.
    \end{align*}
    By the law of total expectation, we know that
    \[
        \E[M]=\E[M \cond \ca\cap\cb]\P(\ca\cap\cb)+\E[M \cond \ca^c\cup\cb^c]\P(\ca^c\cup\cb^c).
    \]
    In the following, we analyze the two conditional expectations. We notice that all the queries made by the \textsc{NoisyMax} algorithm are spent on the two calls of the \stourmax algorithm in Lines 5 and 10, and the $n-|\cs|$ calls of the \textsc{NoisyCompare} algorithm in Line 8. 
    By Lemma \ref{lemma:tournament}, we know that the expected number number of queries spent on Lines 5 and 10 are $O\left(\E\left[|\cs|\right]\log(1/\delta)\right)$ and $O\left(\E\left[|\ct|\right]\log(1/\delta)\right)$ respectively, because $\delta=o(1)$ and $p$ is a constant bounded away from $0$ and $1/2$. By Lemma~\ref{lemma:noisy-compare}, we have that the expected number of queries spent on Line 8 is at most $(1+o(1))\frac{(n-\E\left[|\cs|\right])\log(1/\delta)}{\DKL}$. Therefore, by the linearity of expectation, we have that 
    \[
        \E[M] = O\left(\E\left[|\cs|\right]\log(1/\delta)\right) + O\left(\E\left[|\ct|\right]\log(1/\delta)\right) + (1+o(1))\frac{(n-\E\left[|\cs|\right])\log(1/\delta)}{\DKL}.
    \]
    
    Notice that on the event $\ca \cap \cb$, we have that $\E\left[|\cs| \cond \ca \cap \cb \right]=o(n)$ and $\E\left[|\ct|\cond \ca \cap \cb\right]=o(n)$ (since $\delta=o(1)$). It follows that
    \[
        \E\left[M \cond \ca\cap\cb\right ]\le (1+o(1))\frac{n\log(1/\delta)}{\DKL}.
    \]
    In the case when $\ca$ or $\cb$ does not hold, we still have that $\E\left[|\cs|\cond \ca^c \cup \cb^c \right]=O(n)$ and $\E\left[|\ct|\cond \ca^c \cup \cb^c\right]=O(n)$. Thus, we have that
    \[
        \E\left[M \cond \ca^c\cup\cb^c\right]=O\left(\frac{n\log(1/\delta)}{\DKL}\right).
    \]
    It follows that 
    \begin{align*}
        \E[M]&\le \E\left[M \cond\ca\cap\cb\right]+\E\left[M \cond \ca^c\cup\cb^c\right] \P(\ca^c\cup\cb^c)\\
        &\le (1+o(1))\frac{n\log(1/\delta)}{\DKL}+O\left(\frac{n\log(1/\delta)}{\DKL}\right)\P(\ca^c\cup\cb^c).
    \end{align*}
    Hence, to complete the proof, it suffices to show that $\P(\ca^c\cup\cb^c)=o(1)$. This is implied by the following two lemmas, whose proofs are deferred to Appendix~\ref{app:helper-max-upper}.
    
    \begin{lemma} \label{lemma:event-A}
        The probability of event $\ca^c$ satisfies that
        \[
            \P(\ca^c)=\co(\exp(-n^{c_1}))
        \]
        for some positive constant $c_1$.
    \end{lemma}
    
    \begin{lemma} \label{lemma:event-B}
        Let $\ce_1 = \{\hat{y}\neq y^*\}$ denote the event that the \stourmax algorithm returns an incorrect estimate in Line 5 of Algorithm~\ref{alg:noisymax}. The conditional probability of event $\cb^c$ given event $\ce_1^c$ satisfies that 
        \[
            \P(\cb^c\cond \ce_1^c)=\co(\exp(-n^{c_2}))
        \]
        for some positive constant $c_2$.
    \end{lemma}
    
    \noindent Together with Lemma~\ref{lemma:tournament}, Lemmas~\ref{lemma:event-A} and~\ref{lemma:event-B} imply that 
    \[
        \P(\ca^c\cup\cb^c)\le \P(\ca^c)+\P(\cb^c)\le \P(\ca^c)+\P(\cb^c \cond \ce_1^c)+ \P(\ce_1)\le \co(\exp(-n^{c_1}))+\co(\exp(-n^{c_2})) + \delta=o(1),
    \]
    which completes the proof of Lemma~\ref{lemma:number-queries-max}. Along with Lemma~\ref{lemma:error-prob-max}, this implies the result stated in Theorem~\ref{thm:max-upper}.
\end{proof}

\section{Extension to Varying $p$} \label{sec:extension}
In our main results for the \orn problem and the \maxn problem, we established tight achievability and converse results under the assumptions that $\delta=o(1)$ and $p$ is a constant bounded away from $0$ and $1/2$. 
% In this section, we still consider the case of vanishing error probability $\delta=o(1)$, but allow $p$ to converge to $0$ or $1/2$ as $n\rightarrow\infty$. We discuss the two problems under such varying $p$, and show that tight achievability and converse results can still be established \emph{except} for the case of $\frac{\log(1/\delta)}{\log(1/p)}=\Theta(1)$. In the following discussion, we focus on the \orn problem. The argument for the \maxn problem would follow similarly.
In this section, we maintain the premise of a diminishing error probability, $\delta=o(1)$, but explore scenarios where $p$ converges to either zero or half as $n$ approaches infinity. Within this context, we delve into the two problems while demonstrating that we can still derive matching achievability and converse results—except in the specific instance where $\frac{\log(1/\delta)}{\log(1/p)}=\Theta(1)$. Our subsequent discussion will primarily revolve around the \orn problem, with a parallel argument applicable to the \maxn problem.

\paragraph{$p\rightarrow 1/2$ as $n\rightarrow\infty$} In this case, we can still show that $(1\pm o(1))\frac{n\log{1/\delta}}{\DKL}$ queries in expectation are both necessary and sufficient for achieving a vanishing error probability $\delta$. Firstly, the proof of Theorem~\ref{thm:or-lower} does not use the constant assumption on $p$, so the lower bound $(1- o(1))\frac{n\log{1/\delta}}{\DKL}$ still holds for any $p$. Secondly, at most $(1+ o(1))\frac{n\log{1/\delta}}{\DKL}$ in expectation can still be achieved by the proposed \sno algorithm. To see this, we focus on the proof of Lemma~\ref{lemma:number-queries-or}, where we analyze the expected number of queries of the algorithm. Recall that in this proof, we bounded the number of queries spent on Lines 3 and 11 using Lemmas~\ref{lemma:check-bit} and~\ref{lemma:tournament-or} respectively. These two lemmas hold for any $p$, and in the case of $p\rightarrow 1/2$, it is not hard to see the the number of queries spent on Lines 3 and 11 are $(1+o(1))\frac{n\log(1/\delta)}{\DKL}$ and $o\left(\frac{n\log(1/\delta)}{\DKL}\right)$ respectively. Therefore, the total expected number of queries is at most $(1+ o(1))\frac{n\log{1/\delta}}{\DKL}$.

\paragraph{$p\rightarrow 0$ as $n\rightarrow\infty$} In the case of vanishing $p$, we will discuss three sub-cases separately. 

Firstly, suppose that $\frac{\log(1/\delta)}{\log(1/p)}=\omega(1)$. In this case, we still have that $(1\pm o(1))\frac{n\log{1/\delta}}{\DKL}$ queries in expectation are both necessary and sufficient for achieving a vanishing error probability $\delta$. The lower bound follows because Theorem~\ref{thm:or-lower} holds for any $p$. The upper bounded is achieved by the proposed  \sno algorithm for the same reason as the case $p\rightarrow 1/2$.

Secondly, suppose that $\frac{\log(1/\delta)}{\log(1/p)}=o(1)$. In this case, at least $n$ queries are necessary for achieving a vanishing error probability $\delta$, while the proposed \sno algorithm uses at most $(1+o(1))n$ queries in expectation. To see the lower bound, we notice that trivially each of the $n$ bits must be queried at least once to achieve any error probability $\delta<1/2$. To see the upper bound, we use Lemmas~\ref{lemma:check-bit} and~\ref{lemma:tournament-or} to bound the number of queries spent on Lines 3 and 11 of the \sno algorithm respectively.  By Lemma~\ref{lemma:check-bit}, we know that the expected number of queries spent on Line 3 is at most
\[
     \frac{n}{1-2p}\left\lceil\frac{\log\frac{1-\delta}{\delta}}{\log\frac{1-p}{p}}\right\rceil=\frac{n}{1-2p}=(1+o(1))n.
\]
By Lemma~\ref{lemma:tournament-or}, we know that the expected number of queries spent on Line 11 is at most
\[
     C\left(\frac{\E[|\mathcal{S}|]}{1-H(p)}+\frac{\E[|\mathcal{S}|]\log(1/\delta)}{\DKL}\right)=O(\E[|\mathcal{S}|])=o(n),
\]
where $\mathcal{S}$ denote the set of elements in ${\bf y}$, and the first equality follows because $1-H(p)=\Theta(1)$ and $\frac{\log(1/\delta)}{\DKL}=\Theta(\frac{\log(1/\delta)}{\log(1/p)})$. Therefore the expected number of queries made by \sno algorithm is at most $(1+o(1))n$.

Finally, suppose $\frac{\log(1/\delta)}{\log(1/p)}=\Theta(1)$. In this case, the expected number of queries for achieving error probability $\delta$ is at least $\max\{n,(1-o(1))\frac{n\log(1/\delta)}{\DKL}\}=\Theta(n)$, while the expected number of queries of the proposed \sno algorithm is at most $\frac{(1+o(1))n}{1-2p}\left\lceil\frac{\log\frac{1-\delta}{\delta}}{\log\frac{1-p}{p}}\right\rceil=\Theta(n)$. So our upper and lower bounds are tight up to constant, but the exact constant is still unknown. This untightness is because when $\frac{\log(1/\delta)}{\log(1/p)}$ is a constant, the $\left\lceil\frac{\log\frac{1-\delta}{\delta}}{\log\frac{1-p}{p}}\right\rceil$ term in Lemma~\ref{lemma:check-bit} cannot be upper bounded by $(1+o(1))\frac{\log\frac{1-\delta}{\delta}}{\log\frac{1-p}{p}}$. How to get the exact constant in this case is still an open problem.

\section*{Acknowledgements}
This work was supported in part by the NSERC Discovery Grant No. RGPIN-2019-05448, the NSERC Collaborative Research and Development Grant CRDPJ 54367619, NSF Grant IIS-1901252 and NSF Grant CCF-1909499.

\bibliography{ref}

\appendices
\section{Useful Lemmas}\label{app:lemma}
Here, we introduce several important lemmas.

\begin{lemma}[Le Cam's Two Point Lemma, see e.g.~\citep{yu1997assouad}]\label{lem:lecam}
    Let $\Theta$ and $\hat{\Theta}$ be arbitrary parameter spaces, and let $L: \Theta \times \hat{\Theta} \to  \mathbb{R}$ be any loss function. Suppose that $\theta_0,\theta_1\in\Theta$ satisfy the following separation condition:
    \begin{align*}
    \forall {\hat{\theta}\in \hat{\Theta}}, L(\theta_0, \hat{\theta})+L(\theta_1, \hat{\theta})\geq \Delta>0.
    \end{align*}
    Then we have
    \begin{align*}
        \inf_{\hat{\theta}} \, \sup_\theta \, \mathbb{E}_\theta[L(\theta,\hat{\theta})]\geq \frac{\Delta}{2}\left(1-\mathsf{TV}(\mathbb{P}_{\theta_0},\mathbb{P}_{\theta_1})\right),
    \end{align*}
    where $\mathbb{P}_{\theta_i}$, $i\in \{0,1\}$, denotes the sample distribution under the parameter $\theta_i$.
\end{lemma}

\begin{lemma}[Bretagnolle–Huber inequality
 (\citep{bretagnolle79}, and Lemma 2.6 in~\citep{tsybakov2004introduction})] \label{lem:bh}
    For any two distributions $\mathbb{P}_1,\mathbb{P}_2$, one has
    \begin{align*}
        \mathsf{TV}(\mathbb{P}_1, \mathbb{P}_2)\leq 1-\frac{1}{2} \exp\left(-D_{\mathsf{KL}}(\mathbb{P}_1, \mathbb{P}_2)\right).
    \end{align*}
\end{lemma}

\begin{lemma}[Divergence Decomposition (\citep{auer1995gambling}, and Lemma 15.1 in~\citet{lattimore2020bandit})]\label{lem:div}
Let $M_i$ be the  random variable denoting the number of times experiment $i\in[n]$ is performed under some policy $\pi$, then for two distributions $\mathbb{P}^\pi$, $\mathbb{Q}^\pi$ under policy $\pi$,
\begin{align*}
    D_{\mathsf{KL}}(\mathbb{P}^\pi, \mathbb{Q}^\pi) = \sum_{i\in[n]}\mathbb{E}_{\mathbb{P}^\pi}[M_i] D_{\mathsf{KL}}(\mathbb{P}_i^\pi, \mathbb{Q}_i^\pi).
\end{align*}
\end{lemma}

\begin{lemma}[Chernoff bound~\citep{mitzenmacher_upfal_2005}] \label{lemma:chernoff}
    Let $Z\sim \mathrm{Binom}(n,p)$. For all $0 < \delta < 1$ and $\epsilon > 0$, it holds that
    \begin{equation*}
        \P\left(Z\ge (1+\delta)np\right)\le \exp\left(-\frac{\delta^2np}{3}\right),
    \end{equation*}
    and
    \[
        \P\left(Z\ge n(p+\epsilon)\right)\le \exp\big( -nD_{\mathsf{KL}}(p+\epsilon \| p)\big).
    \]
\end{lemma}

\section{Helper Lemmas for Theorem~\ref{thm:max-upper}} \label{app:helper-max-upper}
In this appendix, we prove Lemmas~\ref{lemma:event-A} and~\ref{lemma:event-B}, which were used to complete the proof of Theorem~\ref{thm:max-upper}.

\begin{proof}[Proof of Lemma~\ref{lemma:event-A}]
    Notice that
    \begin{align*}
        \P(\ca^c)&=\P\left(\mathrm{Binom}\left(n,\frac{1}{\log n}\right)> \frac{n}{\log n}+n^{3/4}\right)\\
        &\stackrel{(a)}{\le} \exp\left(-\frac{n^{1/2}\log n}{3}\right),
    \end{align*}
    where $(a)$ follows from the multiplicative Chernoff bound (Lemma~\ref{lemma:chernoff} in Appendix~\ref{app:lemma}). Hence, $\P(\ca^c) = \co(\exp(-n^{c_1}))$ for any constant $c_1 < 1/2$.
\end{proof}

\begin{proof}[Proof of Lemma~\ref{lemma:event-B}]
    Notice that on the event $\ce_1^c$, we have that $\hat{y}=y^*$. Define 
    \begin{align*}
        \cx^+ &\triangleq \{x\in \cs^c: x>y^*\},\\
        \cx^- &\triangleq \{x\in \cs^c: x<y^*\}.
    \end{align*}
    Notice that there are two types of elements in the set $\ct$. The first type consists of the elements $x\in \cx^+$ appended into $\ct$ because \snc$(y^*,x,\delta)$ correctly outputs true. The other type consists of the elements $x\in \cx^-$ appended into $\ct$ because \snc$(y^*,x,\delta)$ incorrectly returns true. By Lemma~\ref{lemma:noisy-compare}, the \snc algorithm returns an incorrect estimate with probability at most $\delta$. Thus, $|\ct|$ is statistically dominated by the random variable $|\cx^+|+\mathrm{Binom}(|\cx^-|,\delta)$. Next, we derive large deviation bounds for $|\cx^+|$ and $\mathrm{Binom}(|\cx^-|,\delta)$. First, notice that
    \begin{align*}
        \P(|\cx^+|\ge \sqrt{n})&\stackrel{(a)}{=}\left(1-\frac{1}{\log n}\right)^{\sqrt{n}}\\
        &\stackrel{(b)}{\le} \exp\left(-\frac{\sqrt{n}}{\log n}\right),
    \end{align*}
    where (a) follows because the event $\{|\cx^+|\ge \sqrt{n}\}$ is equivalent to the event that the largest $\sqrt{n}$ elements are not sampled into $\cs$, and (b) follows by $1+x\le \exp(x)$. Next, we have that
    \begin{align*}
        \P(\mathrm{Binom}(|\cx^-|,\delta)\ge \delta n+n^{3/4})&\le \P(\mathrm{Binom}(n,\delta)\ge \delta n+n^{3/4})\\
        &\stackrel{(c)}{\le}\exp\left(-\frac{\sqrt{n}}{3}\right),
        % &\stackrel{(c)}{\le}\exp\left(-\frac{(n^{-1/4}/\delta)^2\delta n}{2+n^{-1/4}/\delta}\right)\\
        % &=\exp\left(-\frac{n^{1/2}/\delta}{2+n^{-1/4}/\delta}\right)\\
        % &=\exp\left(-\frac{n^{1/2}}{2\delta+n^{-1/4}}\right),
    \end{align*}
    where $(c)$ follows by the multiplicative Chernoff bound (Lemma~\ref{lemma:chernoff} in Appendix~\ref{app:lemma}). Finally, the proof is completed by realizing that
    \begin{align*}
        \P(\cb^c \cond \ce_1^c)&\le \P(|\cx^+|+\mathrm{Binom}(|\cx^-|,\delta)\ge \sqrt{n}+\delta n+n^{3/4})\\
        &\le \P\left( \{|\cx^+|\ge \sqrt{n}\} \cup \{\mathrm{Binom}(|\cx^-|,\delta)\ge \delta n+n^{3/4}\}\right)\\
        &\le \P(|\cx^+|\ge \sqrt{n})+\P(\mathrm{Binom}(|\cx^-|,\delta)\ge \delta n+n^{3/4})\\
        &\le \exp\left(-\frac{\sqrt{n}}{\log n}\right)+\exp\left(-\frac{\sqrt{n}}{3}\right),
    \end{align*}
    which implies that $\P(\cb^c \cond \ce_1^c) = \co(\exp(-n^{c_2}))$ for any constant $c_2 < 1/2$.
\end{proof}

\end{document}